\documentclass[%
 reprint,
 amsmath,amssymb,
 aps,
prl
]{revtex4-2}

\usepackage{graphicx}
\usepackage{dcolumn}
\usepackage{bm}
\usepackage{xcolor}
\usepackage{physics}
\usepackage{float}

\usepackage{amsthm}
\newtheorem{theorem}{Theorem}

\newtheorem{corollary}[theorem]{Corollary}

\newtheorem{lemma}[theorem]{Lemma}
\newtheorem*{lemma*}{Lemma}

\theoremstyle{definition}

\DeclareMathOperator{\poly}{poly}

\begin{document}
\setlength{\abovedisplayskip}{3pt}
\setlength{\belowdisplayskip}{3pt}

\preprint{APS/123-QED}

\newcommand{\aqa}{\textsuperscript{1}$\langle aQa ^L\rangle $ Applied Quantum Algorithms, Universiteit Leiden}
\newcommand{\liacs}{\textsuperscript{2}LIACS, Universiteit Leiden, Niels Bohrweg 1, 2333 CA Leiden, Netherlands}
\newcommand{\lorentz}{\textsuperscript{3}Instituut-Lorentz, Universiteit Leiden, Niels Bohrweg 2, 2333 CA Leiden, Netherlands}

\title{All this for one qubit? Bounds on local circuit cutting schemes
}

\author{Simon C. Marshall\textsuperscript{1,2}}
\email{s.c.marshall@liacs.leidenuniv.nl}
\author{Jordi Tura\textsuperscript{1,3}}
\author{Vedran Dunjko\textsuperscript{1,2}}
\affiliation{\aqa}
\affiliation{\liacs}
\affiliation{\lorentz}
\newcommand{\sect}[1]{\emph{#1.---}}


\begin{abstract}

Small numbers of qubits are one of the primary constraints on the near-term deployment of advantageous quantum computing. 
To mitigate this constraint, techniques have been developed to break up a large quantum computation into smaller computations. While this work is sometimes called circuit knitting or divide and quantum we generically refer to it as circuit cutting (CC). Much of the existing work has focused on the development of more efficient circuit cutting schemes, leaving open questions on the limits of what theoretically optimal schemes can achieve.
We develop bounds by breaking up possible approaches into two distinct regimes: the first, where the input state and measurement are fixed and known, and the second, which requires a given cutting to work for a complete basis of input states and measurements. For the first case, it is easy to see that bounds addressing the efficiency of any approaches to circuit cutting amount to resolving \textbf{\textsc{BPP$\stackrel{?}{=}$BQP}}. 
We therefore restrict ourselves to a simpler question, asking what \textit{locally-acting} circuit cutting schemes can achieve, a technical restriction which still includes all existing circuit cutting schemes. In our first case we show that the existence of a locally-acting circuit cutting scheme which could efficiently partition even a single qubit from the rest of a circuit would imply \textbf{\textsc{BPP$=$BQP}}. In our second case, we obtain more general results, showing inefficiency unconditionally. We also show that any (local or otherwise) circuit cutting scheme cannot function by only applying unital channels.

\end{abstract}

\maketitle

\newpage

\sect{Introduction}
The near-term deployment of advantageous quantum algorithms is currently constrained by available hardware.
Among other limitations, modern machines simply do not have enough qubits for the most interesting algorithms. 
In an attempt to augment modern machines, several cutting schemes \cite{bravyi2016trading, Peng_2020, mitarai2021constructing, Eddins_2022, piveteau2022circuit, MarshallQ, lowe2022fast, Optimal2023, harada2023optimal, pednault2023alternative} have been proposed that partition a given quantum circuit into smaller blocks. 
Each block can be run independently and then combined to simulate the output of the original circuit. 
We refer to these techniques collectively as ``Circuit Cutting" (CC).

While the apparent value of these schemes is clear, their practical value is limited by their runtimes
\color{black}
as the number of circuit evaluations required grows exponentially with
each two-qubit gate between partitioned blocks.
\color{black}
For many algorithms, each qubit requires a polynomial number of two-qubit gates with the rest of the circuit, preventing useful application of CC to these cases. Applications are instead relegated to a secondary role, such as augmenting connectivity by adding virtual connections \cite{mitarai2021constructing}. Further use of these techniques depends on how the number of terms required can be reduced.

\begin{figure}
    \centering
    \includegraphics[width=8.6cm]{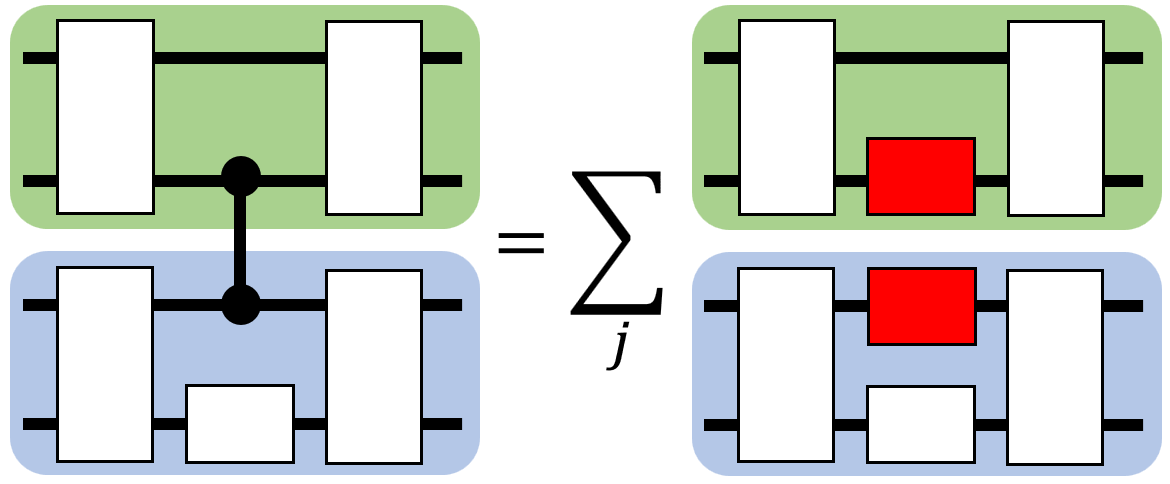}
    \caption{Circuit cutting schemes can be used to simulate computations with more qubits than a user has access to. 
    Here we show a circuit cutting scheme being deployed on a four-qubit circuit, by representing a two-qubit gate as a sum of single-qubit gates we can represent the four-qubit circuit as a sum of tensor products of two-qubit circuits, those two-qubit circuits can then be evaluated on a smaller machine.
    This manuscript focuses only on ``cut-local" schemes, i.e. ones that only modify the circuit at the sites of partition-crossing two-qubit gates, replacing the gates with a sum of single-qubit unitaries.}
    \label{fig:basiccircuit cutting}
\end{figure}

Clearly, some partitions will require a super-polynomial number of circuit evaluations in the worst case if \textbf{\textsc{BPP$\neq$BQP}} (simple example in footnote \footnote{Cutting the circuit in half, then repeatedly cutting the subcircuits generated by this cut in half would require only $\lceil{(log(n))}\rceil$ rounds to reduce an $n$-qubit circuit to a combination of 1-qubit circuits. Each of these 1-qubit circuits can be simulated classically if each round produces at most $A$ subcircuits only $A^{log(n)}$ circuit evaluations are needed (and a similar amount of postprocessing). If the cutting procedure produced less than some pseudo-polynomial number of terms each time, $A$, then a classic simulator exists using pseudopolynomial time}).
This super-polynomial requirement would be more achievable if it scaled in a different parameter, such as the number of qubits in the blocks of the partition.
In some ways, the size of the smallest block is a more natural scaling parameter, for instance, the limited Hilbert space dimension limits the maximum amount of entanglement between blocks (a key ingredient in quantum advantage).

The size of the Hilbert space is a basis for the proofs of circuit cutting scaling requirements in \cite{piveteau2022circuit, lowe2022fast},
where it is shown that an exponential number of terms are needed to simulate large and highly entangled states.
It remains open if a scheme scaling in the size of the smallest block might be possible, as the entanglement, in this case, is highly limited.

\color{black}

This manuscript investigates if there could exist such a circuit cutting scheme, one whose computational overhead would scale polynomially with the number of inter-partition gates, at the cost of an exponential scaling in the size of the smallest partitioned block. We show that when the circuit cutting scheme is limited to local modifications (i.e. it can only modify partition-crossing two-qubit gates, removing this assumption would make any formal statements dramatically more difficult to prove \footnote{As we discuss later, local schemes effectively mean we do not allow significant semantic-preserving circuit rewritings; allowing circuit rewritings is more powerful, but also leads to the (NP-hard) problems of finding minimal or otherwise simplest circuits, making proofs exceptionally difficult. Hence we focus on the simpler case here.}) any circuit cutting scheme that can efficiently remove a single qubit (create a (1, $(n-1)$) partition) would imply \textbf{\textsc{BPP$=$BQP}}.

\sect{Background}
Circuit cutting (CC) schemes \cite{bravyi2016trading, Peng_2020, mitarai2021constructing, Eddins_2022, piveteau2022circuit, MarshallQ, lowe2022fast, Optimal2023,harada2023optimal, pednault2023alternative} (sometimes referred to as circuit partitioning or circuit knitting) are a class of methods designed to reduce the demands on a quantum computer when trying to implement a large quantum circuit. Existing schemes either make multiple calls to the device \cite{bravyi2016trading, Peng_2020, mitarai2021constructing, Eddins_2022} or link multiple devices with classical communication \cite{piveteau2022circuit, lowe2022fast} to simulate the larger circuit. 

Each of these circuit cutting schemes works slightly differently, for simplicity we focus on a generalisation of the formalism presented in \cite{bravyi2016trading}: expressing a given $n$-qubit unitary as a sum of tensor products of two fewer-qubit unitaries,

\begin{equation} 
\label{eq: unitary sum}
U = \sum_i^L \alpha_i U_i'\otimes U_i'',
\end{equation}

\begin{figure}
    \centering
    \includegraphics[width=8.6cm]{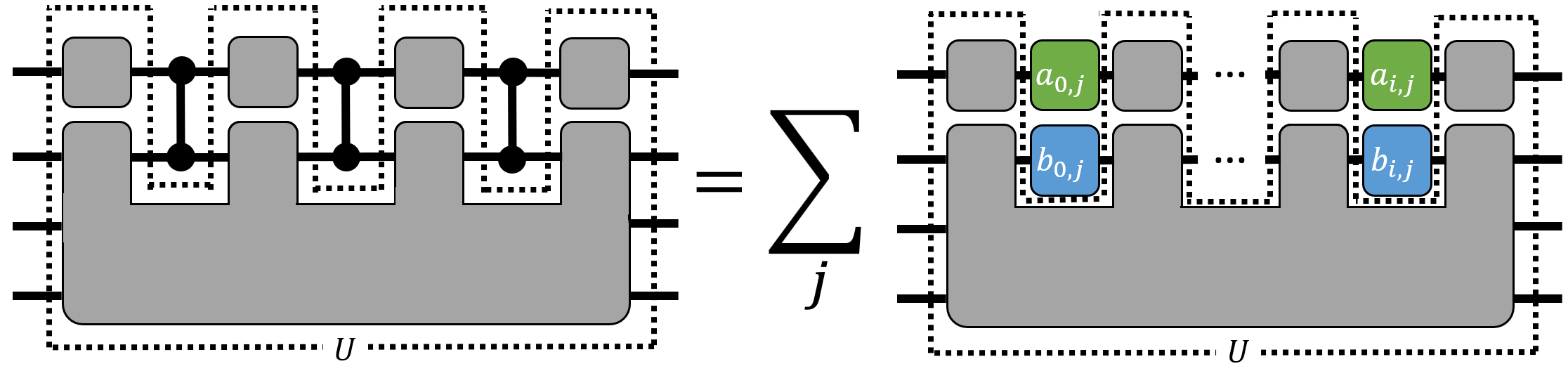}
    \caption{Figure depicts a cut-local circuit cutting scheme being applied to our quantum comb formalism. 
    The comb is designed to specify the whole circuit except for some ``gaps", which are then filled by input gates, in this case, Controlled-Z gates, that can add connectivity between the first qubit and the other $(n-1)$-qubits.
    When a cut-local circuit cutting scheme is applied to the circuit, the connecting two-qubit gates in the gaps become single-qubit operators.
    The resulting unitary is now a tensor-product allowing the first qubit and the other $(n-1)$-qubits to be run separately, reducing the required width of quantum computer. To go beyond cut-local schemes, one would allow $U$ to depend on $j$.}
    \label{fig:comb}
\end{figure}

\noindent but extensions of our results to other schemes (that decompose e.g. tensor-networks \cite{Peng_2020} or superoperators \cite{mitarai2021constructing}) are covered in the ``Generalisation to other schemes'' section.

The existing approaches to circuit cutting incur a super-polynomial scaling in the number of connections, $k$, crossing the partition, i.e. $L \geq c^k$ for some $c\geq 2$.
This manuscript asks if there could exist a scheme capable of partitioning an $n$-qubit circuit into an $m$-qubit block and an $(n-m)$-qubit block with super-polynomial scaling only in $m$ while remaining polynomial in $k$ and $n$, i.e. $L \in O(c'^m \times \poly(n,k))$. We find that even considering the simplest case, where $m=1$, is sufficient to show it is not possible.

 There are conflicting intuitions behind the possible existence of a $\poly(n,k)$ scheme for the $m=1$ case.
On one hand, the addition of a qubit doubles the dimension of the relevant Hilbert space, making it unclear how to simulate the larger Hilbert space with access to only the smaller one.
On the other hand, existing limitations of circuit cutting rely on simulating states with large amounts of entanglement between the blocks to show that $L \geq 2^{k}$. However, in the $m=1$ case, this entanglement is heavily limited, breaking the assumptions behind these limitations. 
Indeed if we simplify the goal to expressing equation \ref{eq: unitary sum} as a sum of \textit{arbitrary linear operators}, it clearly only requires 4 terms to satisfy equation \ref{eq: unitary sum} 
\footnote{\color{black}
This is the upper bound on the Schmidt rank when partitioning a 2-dimensional subspace. To see this, express $U$ in the sum of Pauli strings (which form a basis) and then sum the operators in the $(n-1)$-qubit partition:
$\sum_i \alpha_i \bigotimes_j P_{i,j} = 
\sum_i \alpha_i P_0 \otimes \bigotimes_j P_{i,j} =  \newline
\sum_{p\in\{X,Y,I,Z\}} P_p \otimes \sum_i \alpha_{i,p}\bigotimes_j P_{i,j} =
\sum_p P_p \otimes A_p$\color{black}}.
This manuscript resolves this uncertainty by showing that when the circuit cutting scheme is only allowed to make local modifications an exponential number of terms in $k$ is necessary, regardless of how many qubits are removed.
\color{black}

The existing circuit cutting schemes all share a property we call ``cut-locality": when modifying a gate the resulting subcircuit only differs at the site of the removed gate. In \cite{bravyi2016trading} this is a two-qubit gate replaced locally as the sum of one-qubit gates (see Figure \ref{fig:basiccircuit cutting}).

To formally treat cut-local schemes it is useful to use a variation of the quantum comb formalism \cite{chiribella2008quantum}. We are only interested in using the formalism to partition a single qubit from our circuit, thus we will slightly modify the definition of a quantum comb. 
Informally our quantum comb is a unitary with $G$ ``gaps" where gates can be plugged in, the unitary does not have connections between the first qubit and the rest of the qubits but by putting in two-qubit gates to these gaps we can create a connected unitary. An example is shown in Figure \ref{fig:comb}. We define the quantum comb as a map, $U(\cdot)$, taking in two-qubit unitaries, $G_i$, and returning a fixed $n$-qubit unitary with $G_i$ in the gaps as described. 

Our ultimate goal is to bound cut-local schemes, which would transform a quantum comb with entangling two-qubit gate arguments into the sum of quantum combs with tensor product arguments:

\begin{equation}
U(\ldots, G_j,\ldots) = \sum_{i=0}^L \alpha_i U(\ldots, a_{i,j}\otimes b_{i,j},\ldots)
\label{eq:partitioned Q comb}
\end{equation}

\noindent where $a_{i,j}, b_{i,j} \in SU(2)$ (single qubit unitaries), $\alpha_i \in \mathbb{C}$ (some coefficient) and $G_j\in SU(4)$ (two-qubit gates). 
If a given quantum comb with given two-qubit unitary inputs can be represented as the sum of quantum combs with tensor product inputs in at most $L$ terms, as in equation \ref{eq:partitioned Q comb}, we say there exists an \textit{$L$-term partitioned quantum comb} representation. 

We have yet to define whether the quantum comb is equal to a specific computation, with known input and observable (i.e. classically specified and fixed), or to the unitary itself (so the decomposition does not benefit from considering specific input states or measurements, and must correctly apply to all inputs and measurements). The next section will address both of these cases; we show the unitary case is simple via linear algebra, the fixed-input-output case is more challenging, requiring a complexity-theoretic argument.

\sect{Bounds on the optimal scheme}
Our results are structured into two groups:
The first group, Lemma \ref{thm:exists1} and Theorem \ref{thm:BQP}, shows that when the input and measurement are fixed it is technically possible to partition a single qubit in polynomially many terms, but that if such a decomposition could be found in polynomial time then \textbf{\textsc{BPP$=$BQP}}.
Our second group, centred on Theorem \ref{thm: input non existence} shows that even with unlimited runtime a local circuit cutting scheme cannot find a decomposition of a given circuit in polynomially many terms if we force the same decomposition to apply for every input and observable.

Our first result, showing that a decomposition can be found, is only true in a technical sense and not informative to practical use.

\begin{lemma}\label{thm:exists1}
Given a circuit expressed as a quantum comb with fixed gates, $U(G_1, \ldots)$, a known (i.e. there is a known succinct classical description) input, $\ket{\phi}$ and known observable, $M$, the 
associated expectation value can be expressed with a 1-term partitioned quantum comb in the same expectation value: 

\begin{figure*}
    \centering
    \includegraphics[width=17.8cm]{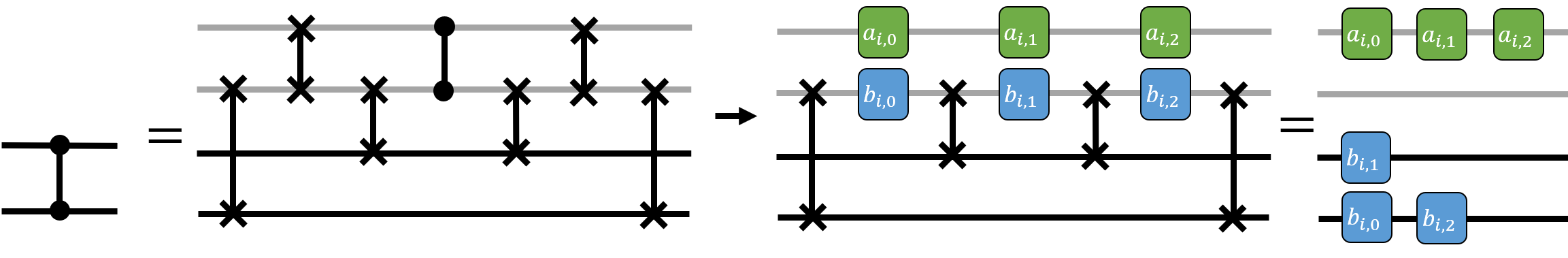}
    \caption{
    Demonstration of a gadget used in proofs of theorems \ref{thm:BQP} and \ref{thm: input non existence}. The gadget adds two ancilla qubits (shown in grey) to move any non-SWAP two-qubit gate to these two ancilla qubits. Applying any local circuit cutting scheme to separately partition the first qubit then leaves a circuit of only SWAP and single-qubit gates. This can then be classically rewritten in polynomial time as a sum of circuits of single-qubit unitaries. To evaluate this circuit classically requires only polynomial time, this is used to reach a contradiction to prove theorem \ref{thm:BQP}. Polynomial sums of single-qubit unitaries can only have polynomial Schmidt rank, this is used to reach a contradiction to prove theorem \ref{thm: input non existence}.}
    \label{fig:gadget}
\end{figure*}

\begin{equation}
    \bra{\phi}U^\dagger(G_1, \ldots)MU(G_1, \ldots)\ket{\phi} = \label{eq: uncut expectation}
\end{equation}
\begin{equation}
    \alpha_0
    \bra{\phi} U^\dagger({a}_{0,0} \otimes {b}_{0,0})M U({a}_{0,0}\otimes {b}_{0,0})\ket{\phi}, \label{eq: cut expectation}
\end{equation}
where $a_{0,0}, b_{0,0} \in SU(2)$ (single qubit unitaries), $\alpha_0 \in \mathbb{C}$ (coefficient).
    
\end{lemma}

\begin{proof}
The outcome of equation \ref{eq: uncut expectation} is equal to some real number, $\gamma$. If $\gamma = 0$, set $\alpha_0$ to 0 and ${a}_{0,0}$ and ${b}_{0,0}$ to any single qubit gates. The equality holds proving this case.

If $\gamma \neq 0$, by \cite{bravyi2016trading} it is known that there must exist some input ${a}_{0,0}$ and ${b}_{0,0}$ that produces a non-zero output of equation \ref{eq: cut expectation}. This non-zero output can then be rescaled with some $\alpha_0 \in \mathbb{C}$ to achieve the equality.
\end{proof}

It is easy to see that if a scheme existed to produce this 1-term partitioned quantum comb, then we could iteratively apply it to the whole circuit and obtain a classical algorithm with polynomial runtime to compute any quantum circuit, implying \textbf{\textsc{BQP}} $=$ \textbf{\textsc{BPP}}. 
The following theorem then demonstrates that \textit{finding} this partition must in some way contain the hardness of \textbf{\textsc{BQP}}, indeed it shows that the existence of any polynomial-time circuit cutting algorithm capable of finding this polynomial termed partitioned quantum comb would imply \textbf{\textsc{BQP}} $=$ \textbf{\textsc{BPP}}.

\begin{theorem}
\label{thm:BQP}
    If there exists a polynomial time classical algorithm that takes an arbitrary input circuit expressed as a quantum comb with fixed gates, $U(G_1, \ldots)$, a known input $\ket{\phi}$ and a known observable $M$, 
    and returns the arguments of an 
    $L$-term partitioned quantum comb, $a_{i,j}, b_{i,j} \in SU(2), \alpha_i \in \mathbb{C}$ for $L \in poly(k, n)$, such that: 
    
    $$
\bra{\phi}U^\dagger(G_1, \ldots)MU(G_1, \ldots)\ket{\phi} =
    $$
    $$
\bra{\phi}\sum_i^L \Bar{\alpha_i} U^\dagger({a}_{i,0} \otimes{b}_{i,0}, \ldots)M\sum_i^L 
\alpha_i U({a}_{i,0} \otimes{b}_{i,0}, \ldots)\ket{\phi}
    $$
    then \textbf{\textsc{BQP}} $=$ \textbf{\textsc{BPP}}.

\end{theorem}

\begin{proof}
Given an algorithm that can partition a single qubit as described (call this algorithm $\mathcal{A}$) we provide an efficient classical simulator.

Given a $n$-qubit input circuit $V$ written in some standard gate set (e.g. H, CNOT, T), replace every existing two-qubit gate anywhere in $V$ with the gadget shown in Figure \ref{fig:gadget}, creating a circuit of only swap gates and single-qubit gates everywhere except for arbitrary 2 qubit gates between the first 2 qubits.

Applying $\mathcal{A}$ to separate the top qubit produces a circuit of only SWAP and single qubit gates, which is classically simulatable in $\poly(n)$ time \cite{nielsen2002quantum}. 
\end{proof}

This represents our main results: put simply, local schemes cannot partition even a single qubit without paying an exponential cost somewhere. Extending this result to the impossibility of separating $l$ qubits is just a matter of ``padding" the gadget with $l-1$ qubits which do not interact with the rest of the circuit. 
Note that this result transfers to the task of approximating (rather than exactly recreating) the output with an $L$-term partitioned quantum comb as \textbf{\textsc{BQP}} is robust to approximations of outputs.
We will also describe how this result can be extended to other local circuit cutting schemes in the next section.

While the condition \textbf{\textsc{BQP}} $\neq$ \textbf{\textsc{BPP}} is a reasonable requirement, it is not clear if it is necessary. We show that by forcing one decomposition to apply for all inputs and measurements (which is equivalent to demanding the unitaries are the same) we can show unconditionally that it is impossible to partition one qubit from an $n$ qubit circuit with only polynomially many terms. 

\begin{theorem}\label{thm: input non existence}
There exists quantum circuits expressible as a quantum comb with input gates, $U(G_1, \ldots)$, such that 
for all $L\in O(poly(n))$ and inputs $\alpha_i \in \mathbb{C}, {a}_{i,j}, {b}_{i,j} \in SU(2)$,

$$
U(G_1, \ldots) \neq
\sum_i^L \alpha_i U({a}_{i,0} \otimes {b}_{i,0}, \ldots).
$$

\end{theorem}
\begin{proof}
    Define $C$ as the circuit that generates $n/2$ Bell pairs from the $\ket{0}^{\otimes n}$ state.
    We can apply the rewriting gadget in Figure \ref{fig:gadget} to $C$, call this new circuit $C'$. 
    
    Assume towards contradiction that there exists a polynomial-$L$ term quantum comb implementing the same unitary as $C'$, by separating the first qubit we have created a sum of $L$ tensor product circuits. As shown in Figure \ref{fig:gadget}, each of these circuits acts locally on every qubit, thus each individual circuit has an operator Schmidt rank of 1 across any partition. 
    Summing over $L$ tensor product terms produces an operator of Schmidt rank at most $L$.
    Applying this circuit to the Schmidt rank 1 input state, $\ket{0}^{\otimes n}$, we produce an output state of Schmidt rank at most $L$, but $n/2$ Bell pairs require a Schmidt rank of atleast $2^{n/2}$  \cite{piveteau2022circuit}, which is a contradiction.
\end{proof}

As with the previous theorem,
the proof of Theorem \ref{thm: input non existence} also extends to the case of \textit{approximating} a unitary easily; 
the fidelity between the closest $poly(n)$-Schmidt rank state and the $n/2$-Bell-pairs state decays exponentially in $n$. 
This implies that the operator distance (and diamond-norm distance) between $U$ and any polynomial sum approximating $U$ also becomes maximal in $n$. 

\sect{Generalisations to other schemes}\label{sect: extensions}
Our results have bounded how any locally acting unitary-based circuit cutting schemes can perform, but we have said relatively little about how a general scheme (one which can express circuits as the sum of other circuits of any form) may perform. 
It is therefore important to determine how broadly our results apply.
In this section, we generalise our framework to encompass other locally acting circuit cutting schemes and discuss how apparently promising routes to generalise to non-local circuit cutting schemes do not work out.

To
generalise our technique to other locally acting circuit cutting schemes note that the choice to decompose a unitary into other unitaries, while useful for illustration, was not maximally general. Instead, we can consider decomposing the \textit{channel} associated to that unitary, $\mathcal{U}$, into other channels:

\begin{equation} 
\label{eq: channel sum}
 \mathcal{U} = \sum_i^L \alpha_i \mathcal{C}_i
\end{equation}

\noindent $\mathcal{C}_i$ now respect an analogous cut locality condition.
If a scheme obeys this locality condition (as \cite{mitarai2021constructing} does) then the gadget can be applied to convert a connected circuit into a tensor product, allowing for classical simulation and extending Theorem \ref{thm:BQP} to this case.
Even classical augmentation of the channel (e.g. with classical communication \cite{piveteau2022circuit, lowe2022fast}) would not break this simulability argument, further extending our results to this case.

It is less clear how the circuit cutting schemes that cut qubits time-wise (i.e. decompose identity channels \cite{Peng_2020, lowe2022fast}) fit into this framework. 
The time-like circuit cutting schemes can be used to create partitioned blocks by cutting qubits that appear in two otherwise disconnected blocks, the choice of which qubits to cut is not immediately clear in our problem (which is to reduce the hardware requirements by just one qubit), instead we must try and find the analogous problem.
If we only allow modifications outside the quantum comb, but still require blocks of at most $(n-1)$-qubits then the only option is to decompose local channels on the 2\textsuperscript{nd} qubit. In this case, our results extend.
\color{black}

Extending these results even further, to non-local (i.e. unrestricted) circuit cutting schemes, generally becomes much more challenging. 
The question now runs into issues of deciding the minimum circuit size necessary to implement a given function, related to the famously opaque minimum circuit size problem and its quantum analogue \cite{chia2021quantum}. \color{black}
Fortunately existing schemes operate using only local cuts, making this question less relevant.
\color{black}

Finally, we wish to address an ostensible link between bounds on non-local circuit cutting and
the one clean qubit model \cite{knill1998power}.
The one clean qubit model is a restricted computational model consisting of an arbitrary circuit taking input of one clean qubit in some fiducial state and $(n-1)$ maximally mixed qubits. 
If one applies the types of circuit cutting methods discussed in this paper to achieve an $(1,n-1)$ partition, one may come to the conclusion that the $(n-1)$-qubit computations will be acting on the maximally mixed states, which is classically simulatable. In this case, if the circuit cutting results in just polynomially many terms, the entire circuit cutting computation would be weakly simulatable, which would imply \textbf{\textsc{PH}}$=$\textbf{\textsc{AM}} \cite{fujii2018impossibility}. 
This would constitute a rather elegant general no-go result for circuit cutting methods achieving a sub-exponential number of terms.
However, the argument fails as circuit cutting does not necessarily apply a sum of just unitary channels to the maximally mixed input (e.g. in \cite{bravyi2016trading} different unitaries might be multiplied to the left and right side of the state, which is not a unitary channel and doesn't preserve the classical simulability of the maximally mixed state) 
or necessarily compute a $(n-1)$-qubit circuit on a subsystem of just the original (maximally mixed) input.
Indeed this argument can be modified to show a slightly more general result: that all circuit cutting schemes must apply non-unital channels (which contain unitary channels), regardless of the number of terms generated (exponential or otherwise).

\begin{corollary}
No circuit cutting scheme can decompose any given unitary, $U$, on a given partition into a finite sum of only  unital channels.
\end{corollary}

\begin{proof}
The proof (provided in appendix B) functions by using two SWAP gates to swap a clean qubit into a maximally mixed block.
\end{proof}

\sect{Summary}
In this manuscript we have analysed the limits of locally acting circuit cutting schemes' ability to partition whole qubits. We found that for polynomially many two-qubit gates between the single qubit and the rest of the circuit, no locally acting scheme can achieve polynomial runtimes. 
We discussed how these results can be extended into other locally acting circuit cutting schemes, such as the superoperator or tensor network formalisms and suggested that they may apply to all local schemes.

This manuscript suggests a clear future research direction: to either generalise these results to non-local circuit cutting schemes, or to attempt to utilise some of the intuitions generated  here to produce an efficient scheme for removing a single qubit from a circuit.

\sect{Acknowledgments}
SCM thanks Angus Lowe and Stefano Polla for interesting and useful discussion.
SCM and VD thank Adri\'an P\'erez-Salinas for useful discussions in the early stages of this work. 
The authors thank Alice Barthe, Kshiti Sneh Rai, Yash Patel, Angus Lowe for useful comments on the manuscript of this work.

VD and SCM acknowledge the support by the project NEASQC funded from the European Union’s Horizon 2020 research and innovation programme (grant agreement No 951821). VD and SCM also acknowledge partial funding by an unrestricted gift from Google Quantum AI. VD was supported by the Dutch Research Council (NWO/OCW), as part of the Quantum Software Consortium programme (project number 024.003.037).
This work has received support from the European Union’s Horizon Europe program through the ERC StG FINE-TEA-SQUAD (Grant No. 101040729). The authors also acknowledge support from the Quantum Delta program. This publication is part of the ‘Quantum Inspire – the Dutch Quantum Computer in the Cloud’ project (with project number [NWA.1292.19.194]) of the NWA research program ‘Research on Routes by Consortia (ORC)’, which is funded by the Netherlands Organization for Scientific Research (NWO).

\bibliography{apssamp}

\section{Appendix A: notes on Theorem 2}

\sect{An alternative way to understand the proof of Theorem \ref{thm:BQP}} We can construct a \textbf{\textsc{BQP}}-complete class of quantum circuits which become classically simulatable if the first qubit is separated.

\sect{Generalising Theorem \ref{thm:BQP}} 
\begin{figure}
    \centering
    \includegraphics[width=8.6cm]{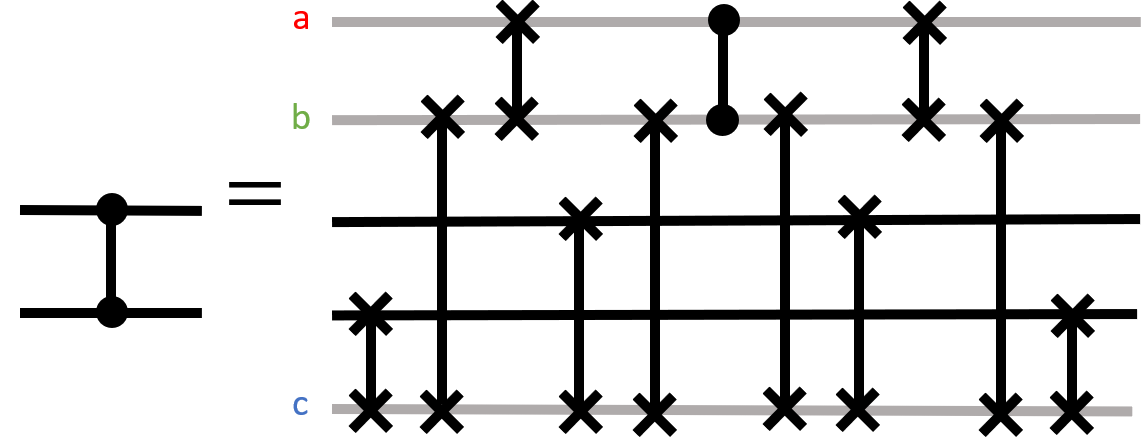}
    \caption{A more advanced gadget that now prevents removal of the first or last qubits.}
    \label{fig:pedants gadget}
\end{figure}
We have not shown a single qubit partition does not exist, we have technically only shown that the first qubit can not be partitioned. This can be remedied by using the slightly more complicated gadget shown in figure \ref{fig:pedants gadget}. This gadget replaces every SWAP gate to ancilla qubit b with a swap gate to ancilla qubit c (at the bottom of the circuit) and then SWAP to ancilla qubit b.
This implements the same operation as a direct SWAP to ancilla b, meaning the removal of the first qubit is still impossible, but now if the last qubit is removed the circuit decomposes to 3 circuits, a circuit consisting of just ancilla c, a circuit consisting of just ancilla qubits a and b, and the rest of the circuit, which is now only SWAPs and single-qubit gates.

\section{Appendix B: Any cutting scheme introduces non-unital channels}
Take some one-clean-qubit unitary, $U$.
Rearrange the circuit form of $U$to interact only twice with the one clean qubit, via two SWAPs:
\begin{figure}[H]
    \centering
    \includegraphics[width=8.6cm]{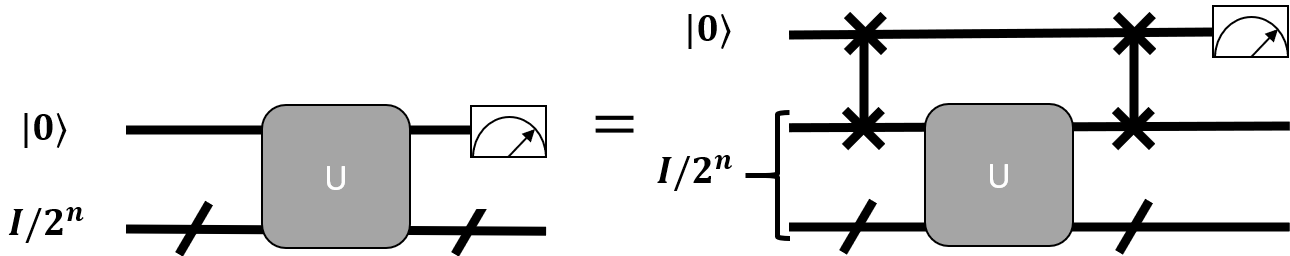}
    \label{fig:DQC1}
\end{figure}
Applying any unital-only-cutting scheme to cut these two qubits results in a sum of channels,
$$
\sum_i^L \alpha_i \mathcal{C}^a_i \otimes \mathcal{C}^b_i 
$$
but by definition, the unital channel applied to the maximally mixed state is the maximally mixed state. Thus the final state is 
$$
(\sum_i^L \alpha_i \mathcal{C}^a_i \ket{0}\ket{0}) \otimes \frac{1}{2^n} \mathbb{I},
$$
which cannot be correct in general, as seen for almost any non-trivial $U$, e.g. n=1, $U=SWAP$.

\end{document}